\documentclass[fleqn,10pt]{wlpeerj} 

\title{Automated Verification of Integer Overflow}

\author[1]{Asankhaya Sharma}
\affil[1]{SourceClear}
\corrauthor[1]{Asankhaya Sharma}{asankhaya@sourceclear.com}

\begin{abstract}
Integer overflow accounts for one of the major source of bugs
in software. Verification systems typically assume a well defined underlying
semantics for various integer operations and do not explicitly 
check for integer overflow in programs. In this paper we present a specification 
mechanism for expressing integer overflow. We develop an
automated procedure for 
integer overflow checking during program verification.
We have implemented a prototype integer overflow checker and tested it
on a benchmark consisting of already verified programs (over 14k LOC).
We have found 43 bugs in these programs due to integer
overflow.
\end{abstract}

\usepackage{ifthen}
\usepackage{times,frameit,epsfig,amsmath,color,latexsym,graphics,wrapfig,clrscode3e,amsthm}
\usepackage{frameit,amsmath,color,latexsym,graphics,wrapfig}

\usepackage{algorithm}
\usepackage{algorithmic}
\usepackage{longtable}
\usepackage {graphicx}
\usepackage{subfigure}
\usepackage{fullpage}
\usepackage{lscape}
\usepackage{wasysym}
\usepackage{stmaryrd}
\usepackage{amssymb}
\usepackage{nameref}
\usepackage{color}   
\usepackage[backref=page]{hyperref}
\hypersetup{
    colorlinks=false, 
    linktoc=all,     
    linkcolor=blue,  
}

\renewcommand*{\cite}{\citep}
\newcommand{\infc}{\ensuremath{{\infty}}}

\newcommand{\smallsep}{\scriptsize \sep}

\newcommand{\conjstar}{%
\mathrel{\ooalign{$\wedge$\kern-5.5pt\smallsep}}}
\newcommand{\sepc}{%
\mathrel{\ooalign{\ensuremath{\fullmoon}$\kern-6.5pt\smallsep$}}}

\newcommand{\rn}{\ensuremath{\Psi}}

\newcommand{\mustF}{\ensuremath{{\mho}_{\scriptsize\it ioc}}}

\newcommand{\sepnode}[3]{\ensuremath{#1{\mapsto}#2(#3)}}
\newcommand{\seppred}[2]{\ensuremath{#1(#2)}}

\newcommand{\self}{\btt{root}}
\newcommand{\report}[1]{ }

\newcommand{\acm}[1]{ }

\newcommand{\hide}[1]{}
\newcommand{\hideie}[1]{}
\newcommand{\nil}{\btt{null}}

\newcommand{\emp}{\btt{emp}}

\newcommand{\veq}{\ensuremath{\equiv}}

\newcommand{\pure}{\ensuremath{\pi}}

\newcommand{\heap}{\ensuremath{\kappa}}

\newcommand{\constr}{\ensuremath{\Phi}}

\newcommand{\entailK}[5]{\ensuremath{#3{\vdash}^{#1}_{#2}#4\,{\sep}\,#5}}

\newcommand{\entailVV}[3]{\entailK{\heap}{V,I}{#1}{#2}{#3}}

\newcommand{\myit}[1]{\textit{#1}}

\newcommand{\view}[2]{\ensuremath{\btt{#1}{\langle}{#2}{\rangle}}}

\def\sep{\code{*}}

\def\primeV{\myit{prime}}

\newcommand{\rulen}[1]{\ensuremath{{\bf \scriptstyle [#1]}}}

\newcommand{\verirulen}[1]{[\underline{{\bf \scriptstyle FV-}\rulen{#1}}]}

\newcommand{\hformn}[3]{\ensuremath{{\btt{#1}}{\mapsto}{\btt{#2}}{\langle}{\btt{#3}}{\rangle}}}
\newcommand{\hformp}[3]{\ensuremath{{\btt{#2}}{\langle}{\btt{#1},\btt{#3}}{\rangle}}}

\newcommand{\arrimp}{\ensuremath{\,\btt{*}\!\!\!\rightarrow}\,}

\def\D{\Delta}

\def\where{\btt{where}}

\def\bool{\code{bool}}
\def\int{\code{int}}
\def\uint{\code{uint}}

\def\float{\code{float}}

\def\void{\code{void}}

\def\pre{\constr_{\myit{pr}}}
\def\post{\constr_{\myit{po}}}

\def\true{\code{true}}
\def\false{\code{false}}
\def\a{a}

\newcommand{\myif}[3]{\code{if}~#1~\code{then}~#2~\code{else}~#3}

\newcommand{\mywhile}[2]{\code{while}~#1~\code{do}~#2}

\newcommand{\code}[1]{{{\ensuremath{\tt #1}}}}
\newcommand{\sm}[1]{{\small \mbox{$#1$}}}
\newcommand{\btt}[1]{{\ensuremath{\tt #1}}}

\newcommand{\passref}{\btt{ref}}
\newcommand{\ensures}{\btt{ensures}}

\newcommand{\requires}{\btt{requires}}
\newtheorem{thm}{Theorem}

\newtheorem{lemma}[thm]{Lemma}

\def\S{\Sigma}

\def\nochange{\myit{nochange}}

\def\sep{\ensuremath{*}}

\newcommand{\hlr}[3]{\ensuremath{\frac{\begin{array}{c}\verirulen{#1}\\[0.5ex]
#2\end{array}}{#3}}}

\newcommand{\htriple}[3]{\ensuremath{\vdash \{#1\}\,#2\,\{#3\}}}

\newcommand{\mymodel}{\models}

\newcommand{\go}[2]{#1\mymodel #2}

\def\inv{\myit{inv}}

\renewcommand{\view}[2]{\ensuremath{#1{\langle}{#2}{\rangle}}}

\newcommand{\atom}{\alpha}

\begin{document}

\flushbottom
\maketitle
\thispagestyle{empty}

\section*{Introduction} \label{sec.intro}
Numerical errors in software are quite common and yet are ignored by most verifications systems.
Integer overflow in particular has been among the top 25 software bugs \cite{Chirstey:SANS11}.
These errors can lead to serious failures and exploitable vulnerabilities.
Program verification helps to build reliable software by ensuring that relevant properties of 
can be validated. In Hoare logic style program verification we typically specify programs using
pre and post conditions. The verifier assumes an underlying well defined semantics and generates 
proofs of correctness of program. Most verification systems do not check for errors due to
undefined behaviors in programs. In C/C++ the integer operations may lead to undefined behaviors
as specified by the standard. Undefined behaviors in the C/C++ specification lead to confusion among
programmers.
Moreover many programmers expect wrap around behavior for integer overflow
\cite{Dietz:ICSE12}
and may intentionally write code that leads to such overflows.
The code written with undefined (in the language specification) intentional overflows
is not guaranteed to be portable and the behavior may depend on the optimizations used in various compilers.

It is no surprise that automated verifiers typically assume a well defined semantics for various integer 
operations. However in order to increase the completeness of verification it is desirable to specify and 
verify integer overflow. The starting point of this work is the HIP/SLEEK verification system 
\cite{OOPSLA11:Chin} based on separation logic. HIP/SLEEK system
can do automated verification of shape, size and bag properties of programs \cite{SCP12:Chin}.
We extend the domain of integers (extended number line) with two 
logical infinite constants \code{\infc} and \code{-\infc}, corresponding to positive
and negative infinity respectively. Even though this kind extension of integers
is common in libraries for programming languages for portability reasons, it is eventually typically mapped
to a fixed value for a particular underlying architecture (32-bit or 64-bit).
In a verification setting we find it better to enrich the underlying specification logic with
these constants and reason with them automatically during entailment. This mechanism
allows us to specify intentional and unintentional integer overflow
in programs.
In particular our key contributions are

\begin{itemize}
 \item A specification mechanism for integer overflows using logical infinities
 \item Entailment procedure for handling logical infinities
 \item Integrated integer overflow checking with automated verification
 \item A prototype implementation of an Integer overflow checker
 \item Finding 43 integer overflow bugs in existing benchmark of verified programs
\end{itemize}

The rest of the paper is structured as follows.
In the next section, we motivate our approach with a few
examples. Then we present our specification language with
logical infinities. This specification language is used to describe automated verification with
integer overflow checking. We also formulate some soundness properties of our system.
In the experiments section we present our implementation with
a benchmark of already verified programs. We describe some related
work and finally we conclude in the last section.


\section*{Motivating Examples}\label{sec:motivate}
We illustrate the integration of integer overflow checking in a verification system by means of few examples.
The following function increments the value passed to it.

\[
 \begin{array}{l}
  \code{void~ex1(int~n)}\\
  \quad \code{requires~n{\geq}0}\\
  \quad \code{ensures~res{=}n{+}1;}\\
  \quad \code{\{}\\
  \quad \quad \code{return~n+1;}\\
  \quad \code{\}}
 \end{array}
\]

If the value of \code{n} that is passed to this function is the maximum value that can be 
represented by the underlying architecture this program will lead to an integer overflow.
In order to avoid dealing with absolute values of maximum and minimum integers we introduce
a logical constant \code{\infc} in the specification language.
With this constant it is possible to write the specification for the function to avoid
integer overflow.

\[
 \begin{array}{l}
  \code{void~ex2(int~n)}\\
  \quad \code{requires~0{\leq}n+1{<}\infc}\\
  \quad \code{ensures~res{=}n{+}1;}\\
  \quad \code{\{}\\
  \quad \quad \code{return~n+1;}\\
  \quad \code{\}}
 \end{array}
\]

Another benefit of adding this constant to the specification language is that it allows
users to specify intentional integer overflow. A recent study \cite{Dietz:ICSE12} has found
intentional integer overflow occurs frequently in real codes. 
We allow a user to express integer overflow using \code{\infc} constant where
such behavior is well defined. The following example shows how to specify intentional overflow.

\[
 \begin{array}{l}
  \code{void~ex3(int~n)}\\
  \quad \code{requires~n{\geq}0}\\
  \quad \code{ensures~(n {+} 1 {<} \infc {\wedge} res{=}n{+}1)}\\
  \quad \code{\qquad \qquad \vee~ (n {+} 1 {\geq} \infc {\wedge} (\true) \mustF );}\\
  \quad \code{\{}\\
  \quad \quad \code{return~n+1;}\\
  \quad \code{\}}
 \end{array}
\]

In this example we use the error calculus from \cite{NFM13:Loc} to
specify integer overflow with an error status (\code{\mustF}). This error status is verified
and propagated during entailment. Details of the error validation
and localization mechanism are given in \cite{NFM13:Loc}.
Use of logical infinity constants (\code{\infc}) enable us to specify intentional integer overflow
as an explicit error scenario in the method specification.
Another benefit of using an enhanced specification mechanism (with \code{\infc} constants)
is that we can integrate integer overflow checking in an expressive logic like
separation logic. This addresses two major problems
faced by static integer overflow checkers - tracking integer overflow through heap
and integer overflows inside data structures. As an example consider the following method 
which returns the sum of values inside a linked list.

\[
 \begin{array}{l}
 \code{data~\{int~val;node~next\}}\\
 \\
 \code{\view{ll}{\self,sum} {\veq}}
 \code{(\self{=}\nil {\wedge} sum{=}0) {\vee}}\\ 
 \quad \code{\,\exists \, d,q \cdot (\hformn{\self}{node}{d,q}{\sep} \hformp{q}{ll}{rest}) 
 {\wedge} sum {=} d {+}rest {\wedge} sum {<} \infc}\\
 \\
  \code{int~ex4(node~x)}\\
  \quad \code{requires~\view{ll}{x,s}}\\
  \quad \code{ensures~\view{ll}{x,s}{\wedge}res{=}s}\\
   \quad \code{\{}\\
  \quad \quad \code{if~(x~==~null)}\\
  \quad \quad \quad \code{return~0;}\\
  \quad \quad \code{else}\\
  \quad \quad \quad \code{return~x.val+ex4(x.next);}\\
  \quad \code{\}}
 \end{array}
\]

We can specify the linked list using a predicate in 
separation logic (\code{ll}). In the predicate definition we use 
\code{\infc} constant to express the fact that the sum of values in the list 
cannot overflow. With this predicate we can now write the pre/post condition for the function \code{ex4}.
During verification we can now check that the sum of values of linked list
will not lead to an integer overflow. A common security vulnerability that can be 
exploited using Integer overflow is the buffer overrun. The following example (\code{ex5})
shows how two character arrays can be concatenated. We express the bounds on the domain of arrays
using a relation \code{dom}. This program can lead to an integer overflow
if the sum of the size of two arrays is greater than the maximum integer that can be represented
by the underlying architecture. By capturing the explicit condition under which 
this function can lead to an integer overflow (\code{\mustF}) we can verify and prevent 
buffer overrun. We again use the \code{\infc} logical constant in the precondition
to specify the integer overflow.

\[
 \begin{array}{l}
 \code{dom(char[]~c,int~low,int~high)}\\
 \\
 \code{dom(c,low,high){\wedge}low\leq l{\wedge}h\leq high{\implies}dom(c,l,h)}\\
 \\
  \code{int~ex5(ref~char[]~buf1,char[]~buf2,size\_t~len1,size\_t~len2)}\\
  \quad \code{requires~dom(buf1,0,len1){\wedge}dom(buf2,0,len2){\wedge}len1+len2\leq256}\\
  \quad \code{ensures~res=0{\wedge}dom(buf1,0,len1+len2);}\\
  \quad \code{requires~dom(buf1,0,len1){\wedge}dom(buf2,0,len2){\wedge}len1+len2>256}\\
  \quad \code{ensures~res=-1{\wedge}dom(buf1,0,len1);}\\
  \quad \code{requires~dom(buf1,0,len1){\wedge}dom(buf2,0,len2){\wedge}len1+len2>\infc}\\
  \quad \code{ensures~(\true) \mustF;}\\
  \quad \code{\{}\\
  \quad \quad \code{char~buf[256];}\\
  \quad \quad \code{if(len1+len2>256)~return-1;}\\
  \quad \quad \code{memcpy(buf,buf1,len1);}\\
  \quad \quad \code{memcpy(buf+len1,buf2,len2);}\\
  \quad \quad \code{buf1 = buf;}\\
  \quad \quad \code{return~0;}\\
  \quad \code{\}}
 \end{array}
\]

Using \code{\infc} constant as part of the specification language we can represent various 
cases of integer overflows in a concise manner. In this way we also avoid multiple
constants like INT\_MAX, INT\_MIN etc., typically found in header files for various architectures.
When compared to other approaches, this specification
mechanism is more suited to finding integer overflows during verification.

Exploitable vulnerabilities caused by integer overflows can also be prevented by
specifications preventing overflow using \code{\infc}. The following example (taken from \cite{ArtSoftBook})
shows how integer overflow checking can prevent network buffer overrun.
For brevity we show only part of the functions and omit the specification for the method as well.
The function \code{ex6} reads an integer from the network and performs some sanity checks on it. 
First, the length is checked to ensure that it's greater than or equal to zero and, therefore, positive. 
Then the length is checked to ensure that it's less than MAXCHARS. 
However, in the second part of the length check, 1 is added to the length. 
This opens a possibility of the following buffer overrun: A value of INT\_MAX passes the first check 
(because it's greater than 0)
and passes the second length check (as INT\_MAX + 1 can wrap around to a negative value). 
read() would then be called with an effectively unbounded length argument,
leading to a potential buffer overflow situation. This situation can be prevented 
by using the specifications given for the \code{network\_get\_int} method
that ensures that length is always less than \code{\infc}.

\[
 \begin{array}{l}
 \code{int~network\_get\_int(int~sockfd)}\\
 \quad \code{requires~\true}\\
 \quad \code{ensures~res<\infc;}\\
  \code{char*~ex6(int~sockfd)}\\
  \quad \code{\{}\\
  \quad \quad \code{char~buf;}\\
  \quad \quad \code{int~length = network\_get\_int(sockfd);}\\
  \quad \quad \code{if(!(buf = (char *)~malloc(MAXCHARS)))}\\
  \quad \quad \quad \code{die(``malloc: \% m '');}\\
  \quad \quad \code{if(length < 0~||~length + 1 >= MAXCHARS)}\\
  \quad \quad \code{\{}\\
  \quad \quad \quad \code{free(buf);}\\
  \quad \quad \quad \code{die(``bad length: \% d '',value);}\\
  \quad \quad \code{\}}\\
  \quad \quad \code{if(read(sockfd, buf, length) <= 0)}\\
  \quad \quad \code{\{}\\
  \quad \quad \quad \code{free(buf);}\\
  \quad \quad \quad \code{die(``read: \% m '');}\\
  \quad \quad \code{\}}\\
  \quad \quad \code{return~buf;}\\
  \quad \code{\}}
 \end{array}
\]
%
%
%
%
%

\section*{Specification Language}\label{sec:sl-infin}
We present the specification language of our system which is extended from 
\cite{SCP12:Chin} with the addition of a constant representing logical infinity. 
The detailed language is depicted in figure \ref{fig.spec_syntax}.
$\pre{\arrimp}\post$ captures a precondition $\pre$ and a postcondition $\post$ of a method or a loop.
They are abbreviated from the standard representation \code{requires} $\pre$ and \code{ensures} $\post$, and
 formalized by separation logic formula  $\constr$.
In turn, the separation logic formula is a disjunction of a heap formula and a pure formula
($\heap{\wedge}\pure$). 
The pure part 
$\pure$ captures a rich constraint from the domains of
Presburger arithmetic (supported by Omega solver \cite{Pugh:CACM}),
monadic set constraint (supported by MONA solver \cite{klarlund-mona})
or polynomial real arithmetic (supported by Redlog solver \cite{Dolzmann:1997:RCA}).
Following the definitions of separation logic  in \cite{Ishtiaq:POPL01,Reynolds:LICS02}, 
the heap part provides notation to denote $\emp$ heap,
 singleton heaps \code{\mapsto}, and disjoint heaps \sep.

The major feature of our system compared to \cite{Ishtiaq:POPL01,Reynolds:LICS02} is the ability
for user to define recursive data structures. Each data structure and its properties
can be defined by an inductive predicate \myit{pred},
that consists of a name \myit{p}, a main
separation formula $\constr$ and an optional pure invariant formula $\pure$
that must hold for every predicate instance. 
In addition to the integer constant $k$ we now support a new 
infinite constant denoted by \code{\infc}. 
This enables us to represent positive and negative infinities by 
\code{\infc} and \code{-\infc} respectively. For the following discussion we assume the existence 
of an entailment prover for separation logic (like \cite{SCP12:Chin})
and a solver for Presburger arithmetic (like \cite{Pugh:CACM}). We now focus only on
integrating automated reasoning with the new infinite constant \code{\infc} inside these existing provers.

\begin{figure}[thb]
\begin{center}
\begin{minipage}{0.6\textwidth}
\begin{frameit}
\[
\begin{array}{ll}
\myit{pred} &::= \seppred{\myit{p}}{v^*}~{\veq}~ \constr ~~ [\code{\inv} ~~\pure]\\
\myit{mspec} &::= \pre{\arrimp}\post\\
\constr &::= \bigvee~(\exists w^*{\cdot}\heap{\wedge}\pure)^* \\
\heap &::= \emp ~|~ \sepnode{v}{c}{v^*} ~|~ \seppred{p}{v^*}
~|~ \heap_1 \sep
\heap_2 \hide{~|~ \heap_1 \mw \heap_2}
\\
\pure &::= \atom ~|~ \pure_1{\wedge}\pure_2
\\
\atom &::= \beta ~|~ \neg \beta\\
\beta &::= 
v_1 {=} v_2 ~|~ v{=}\nil ~|~ a_1{\leq}a_2 ~|~ a_1{=}a_2  
\\
\a &::=~ \!\!\begin{array}[t]{l}
  k \mid k{\times}v \mid \a_1+\a_2 \mid -a
  \mid \max(\a_1,\!\a_2)\\
  \mid \myit{min}(\a_1,\!\a_2)
  \mid \code{\infc} 
      \end{array}
\\
\\
\myit{where} &~  p~\myit{is a predicate name};~ v,w~\myit{are variable names};\\
& ~ c~ \myit{is a data type name};~
k ~ \myit{is an integer
 constant};\\
\end{array}
\]
\caption{The Specification Language}
\label{fig.spec_syntax}
\end{frameit}
\end{minipage}
\end{center}
\end{figure}

An entailment prover for the specification language is used to discharge
proof obligations generated during forward verification.
The entailment checking for separation logic formulas is typically represented \cite{SCP12:Chin} as
follows.

\[
 \constr_1 \vdash \constr_2 , \constr_r
\]

This attempts to prove that
$\constr_1$ entails $\constr_2$ with $\constr_r$ as
its frame (residue) not required for proving $\constr_2$.
This entailment holds, if $\constr_1 \implies \constr_2 \sep \constr_r$.
Entailment provers for separation logic 
deal with the heap part ($\heap$) of the formula and reduce 
the entailment checking to satisfiability queries over the pure part ($\pure$).
We now show how this reasoning can be extended to deal with the new constant
representing infinity (\code{\infc}). A satisfiability check over pure formula with \code{\infc}
is reduced to a satisfiability check over a formula without \code{\infc} which can be 
discharged by using existing solvers (like Omega). In order to eliminate
\code{\infc} from the formula we take help of the equisatisfiable normalization rules shown in
figure \ref{fig.normalize} and proceed as follows.

\[
\begin{array}{l}
SAT(\pure)\\
substitute~equalities(v{=}\code{\infc})\\
\implies SAT([v/\code{\infc}]\pure)\\
normalization \\
\implies SAT(\pure \leadsto \pure_{norm}) \\
elimintate~\code{\infc} \\
\implies SAT([\code{\infc}/v_{\code{\infc}}]\pure_{norm})
\end{array}
\]

We start with substituting any equalities with 
\code{\infc} constants then we apply the normalization rules.
The normalization rules eliminate certain expressions containing 
\code{\infc} based on comparison with integer constants ($k$) 
and variables ($v$). We show rules for both \code{\infc} and \code{-\infc}
in figure \ref{fig.normalize}.
In the normalization rules we use $a_1{\neq}a_2$ as a shorthand for $\neg(a_1{=}a_2)$.
During normalization, we may generate some equalities involving \code{\infc}
(in \rulen{NORM-VAR-INF}).
In that case, we normalize again after substituting the new equalities in the pure formula.
Once no further equalities are generated we eliminate the remaining \code{\infc} constant if any 
by replacing it with a fresh integer variable $v_{\code{\infc}}$ in the pure formula. The pure formula
now does not contain any infinite constants and a satisfiability check on the
formula can now be done using existing methods.

\begin{figure}[thb]
\begin{center}
\begin{minipage}{0.6\textwidth}
\begin{frameit}
 \[
  \begin{array}{c}
 \rulen{NORM-INF-INF}\\
 \code{\infc} {=} \code{\infc} \leadsto \true \\
 \code{\infc} {\neq} \code{\infc} \leadsto \false \\
 \code{\infc} {\leq} \code{\infc} \leadsto \true \\
 \code{\infc} {=} {-\code{\infc}} \leadsto \false \\
 \code{\infc} {\neq} {-\code{\infc}} \leadsto \true \\
 \code{\infc} {\leq} {-\code{\infc}} \leadsto \false \\
 {-\code{\infc}} {=} {-\code{\infc}} \leadsto \true \\
 {-\code{\infc}} {\neq} {-\code{\infc}} \leadsto \false \\
 {-\code{\infc}}{\leq} {-\code{\infc}} \leadsto \true \\
 {-\code{\infc}} {\leq} {\code{\infc}} \leadsto \true 
\end{array}
\qquad
  \begin{array}{c}
  \rulen{NORM-CONST-INF}\\
 k {=} \code{\infc} \leadsto \false \\
 k {\neq} \code{\infc} \leadsto \true \\
 k {\leq} \code{\infc} \leadsto \true \\
 \code{\infc} {\leq} k \leadsto \false \\
 k {=} {-\code{\infc}} \leadsto \false \\
 k {\neq} {-\code{\infc}} \leadsto \true \\
 k {\leq} {-\code{\infc}} \leadsto \false \\
 {-\code{\infc}} {\leq} k \leadsto \true 
\end{array} 
\]
\[
 \begin{array}{c}
  \rulen{NORM-VAR-INF}\\
 v {\leq} \code{\infc} \leadsto \true \\
 \code{\infc} {\leq} v \leadsto v {=} \code{\infc} \\
 v {\leq} {-\code{\infc}} \leadsto v {=} {-\code{\infc}}\\
  {-\code{\infc}} {\leq} v \leadsto \true 
\end{array}
\qquad
 \begin{array}{c}
 \rulen{NORM-MIN-MAX}\\
 min(\a,\code{\infc}) \leadsto \a \\
 max(\a,\code{\infc}) \leadsto \code{\infc} \\
 min(\a,{-\code{\infc}}) \leadsto {-\code{\infc}} \\
 max(\a,{-\code{\infc}}) \leadsto \a 
\end{array}  
 \]
\caption{Equisatisfiable Normalization}
\label{fig.normalize}
\end{frameit}
\end{minipage}
\end{center}
\end{figure}

Enriching the specification language with infinite constants is quite useful as it allows users to
specify properties (integer overflows) using \code{\infc} as demostrated in the motivating examples.
The underlying entailment procedure can automatically 
handle \code{\infc} by equisatisfiable normalization.

\section*{Verification with Integer Overflow}\label{sec:sl-io}
Our core imperative language is presented in figure \ref{fig.syntax}.
A program \code{P} comprises of a list of data structure declarations \code{tdecl^*} and a list of method declarations 
\code{meth^*} (we use the superscript \code{^*} to denote a list of elements).
Data structure declaration can be a simple node \code {datat} or a recursive shape predicate declaration 
\code{pred} as shown in figure \ref{fig.spec_syntax}.

\begin{figure}[ht]
\begin{center}
\begin{minipage}{0.7\textwidth}
\begin{frameit}
\[
\begin{array}{ll}
\myit{P} &::= \myit{tdecl}^* ~\myit{meth}^* \\
\myit{tdecl} &::= \myit{datat} ~|~ \myit{pred} \\
\myit{datat} &::= \btt{data} ~c ~\{~ \myit{field}^* ~\} \\
\myit{field} &::= t~v  \\
t &::= c ~|~ \tau  \\
\tau &::= \uint ~|~ \int ~|~ \bool ~|~ \float ~|~ \void \\
\myit{meth} &::= t ~ \myit{mn}~(([\passref]~t~v)^*)~\where~(\myit{mspec})^*~\{e\} \\
e &::= \nil~|~k^{\tau} ~|~ k_1^{[u]\int} + k_2^{[u]\int} ~|~ v  ~|~ v.f ~|~ v {:=} e ~|~ v_1.f {:=} v_2 \\
   & \quad|~ \btt{new}~c(v^*)  ~|~e_1;e_2 ~|~ t~v;~e
       ~|~ \myit{mn}(v^*)\\
       &\quad |~\myif{v}{e_1}{e_2} ~ |~ \mywhile{v}{e}{(\myit{mspec})^*}
\end{array}
\]
\caption{A Core Imperative Language}\label{fig.syntax}
\end{frameit}
\end{minipage}
\end{center}
\end{figure}

A method is declared with a prototype, its body \code{e}, and multiple specification \code{mspec^*}.
The prototype comprises a method return type, method name and method's formal parameters. 
The parameters can be passed by {\em value} or by {\em reference} with keyword \code{\passref}
and their types can be primitive $\tau$ or user-defined \code{c}.
A method's body consists of a collection of statements.
We provide basic statements for accessing and modifying shared data structures and for explicit allocation
 of heap storage. It includes:
\begin{enumerate}
\item {\em Allocation} statement: \code {new~c(v^*)}
\item {\em Lookup} statement: For simplifying the presentation but without loss of expressiveness, we just provide
 one-level lookup statement \code{v.f} rather than \code{v.f_1.f_2}.
\item {\em Mutation} statement: \code{v_1.f := v_2}
\end{enumerate}

In addition we provide core statements of an imperative language, such as 
semicolon statement \code {e_1;e_2},
 function call \code{mn(v^*)}, conditional statement \code{if~{v}~{e_1}~{e_2}}, and loop statement 
\code {while~{v}~{e}~{(\myit{mspec})^*}}. Note that for simplicity, we just allow boolean variables (but not expression)
 to be used as the test conditions for conditional statements 
and loop statement must be annotated with invariant through \code{mspec^*}.
To illustrate some of the basic operations on integers in the language
we also show the addition operation between two integers (unsigned and signed) in figure \ref{fig.syntax}
as \code{k_1^{[u]\int} + k_2^{[u]\int}}.

We now present the modifications needed to do forward verification with interger overflow.
The core language used by our system is a C-like imperative language described in 
figure \ref{fig.syntax}.
The complete set of forward verification rules are as given in \cite{SCP12:Chin}.
We use $P$ to denote the program being checked. With the pre/post conditions
declared for each method in P, we can now apply modular verification to its body
using Hoare-style triples $\vdash\{\D_1\}e\{\D_2\}$. We expect $\D_1$ to be given before computing
$\D_2$ since the rules are based on a forward verifier. 
To capture proof search,
we generalize the forward rule to the form $\vdash\{\D_1\}e\{\rn\}$ where $\rn$
is a set of heap states, discovered by a search-based verification process \cite{SCP12:Chin}. When $\rn$
is empty, the forward verification is said to have failed for $\D_1$ as prestate.
As most of the forward verification rules 
are standard \cite{Nguyen:VMCAI07}, we only provide 
those for method verification and method call. 
Verification of a method starts with each
precondition, and proves that the corresponding postcondition is
guaranteed at the end of the method. The verification is formalized
in the rule \code{\verirulen{METH}}:

\begin{itemize}
\item function \myit{prime(V)} returns \sm{\{v' \mid v \in V\}}. 
\item predicate \myit{nochange(V)} returns \sm{\bigwedge_{ v {\in} V}
(v=v')}. If \sm{V=\{\}}, \myit{nochange(V)=true}. 
\item \sm{\exists W
\cdot \rn} returns \sm{\{ \exists W \cdot \rn_i | \rn_i \in \rn \}}. 
\end{itemize}

\[
\hlr{METH}{ V{=}\{v_m..v_n\} \quad  \myit{W}{=}\primeV(V) \\
\forall i =1,..,p ~\cdot~(~ \htriple{\pre^i{\wedge}\nochange(V)}{e}{\rn^i_1}  \\
 \qquad \entailVV{(\exists
\myit{W}{\cdot}\rn^i_1)\,}{\post^i}{\rn^i_2} \qquad
\rn^i_2{\neq}\{\}) }{\
t_0~\myit{mn}((\passref~t_j~v_j)_{j{=}1}^{m{-}1},(t_j~v_j)_{j{=}m}^n)~\{\requires~\pre^i~~\ensures~\post^i\}_{i=1}^p~\{e\} }
\]

At a method call, each of the method's precondition is checked, \sm{\entailVV{\D}{\rho\pre^i}{\rn_i}}, where \sm{\rho}
represents a substitution of \sm{v_j} by \sm{v'_j}, for all
\sm{j=1,..,n}. The
combination of the residue \sm{\rn_i} and the postcondition is added to
the poststate. If a precondition is not entailed by the program state
\sm{\D}, the corresponding residue is not added to the set of states.
The test \sm{\rn{\neq}\{\}} ensures that at least one precondition is satisfied.
Note that we use the primed notation for denoting the latest value of a variable.
Correspondingly, 
\sm{[v'_0/v_i]}
is a substitution that replaces  the value \sm{v_i} with the latest value of \sm{v'_0}. 

\[
\hlr{CALL}{
t_0~\myit{mn}((\passref~t_j~v_j)_{j{=}1}^{m{-}1},(t_j~v_j)_{j{=}m}^n)~\{\requires~\pre^i~~\ensures~\post^i\}_{i=1}^p~\{e\} \in \myit{P} \\
 \rho{=}[v'_j/v_j]_{j=m}^n \qquad  \entailVV{\D}{\rho\pre^i}{\rn_i}\quad \forall i {=}1,..,p \\
\rn = \bigcup^p_{i{=}1} \post^i \sep \rn_i
\qquad \rn\,{\neq}\,\{\} } {\htriple{\D}{mn(v_1..v_n)}{\rn}}
\]

In order to integrate integer overflow checking with automated verification we first translate the 
basic operations in the core language (like integer addition)
to method calls to specific functions which do integer overflow checking. In this paper
we illustrate the verification using only the addition overflow, however similar translations can be done
for other operators like multiplication \cite{MSR2009:Moy} etc.. The addition
operation for unsigned integers \code{k_1^{\uint} + k_2^{\uint}} is translated to the method \code{uadd}
whose specification is given below.

\[
 \begin{array}{l}
  \code{int~uadd(uint~k_1,uint~k_2)}\\
  \quad \code{requires~k_1{+}k_2{>}\infc}\\ 
  \quad \code{ensures~(true)\mustF;}\\
  \quad \code{requires~k_1{+}k_2{\leq}\infc}\\
  \quad \code{ensures~res{=}k_1{+}k_2;}\\
 \end{array}
\]

The addition of unsigned integers overflows when their sum is greater
than \code{\infc}. The case of signed integer overflow has several cases. 
We translate addition of signed integers to the method \code{add}.
The specification of the \code{add} method 
covers all the cases for signed integer overflow as detailed in \cite{ISSRE10:Dannenberg}.

\[
 \begin{array}{l}
  \code{int~add(int~k_1,int~k_2)}\\
  \quad \code{requires~k_1{>}0{\wedge}k_2{>}0{\wedge}k_1{+}k_2{>}\infc}\\
  \quad \qquad \qquad \code{\vee k_1{>}0{\wedge}k_2{\leq}0{\wedge}k_1{+}k_2{<}{-\infc}}\\
  \quad \qquad \qquad \code{\vee k_1{\leq}0{\wedge}k_2{>}0{\wedge}k_1{+}k_2{<}{-\infc}}\\
  \quad \qquad \qquad \code{\vee k_1{\leq}0{\wedge}k_2{\leq}0{\wedge}k_1{+}k_2{<}{\infc}{\wedge}k_1{\neq}0}\\
  \quad \code{ensures~(true)\mustF;}\\
  \quad \code{requires~k_1{>}0{\wedge}k_2{>}0{\wedge}k_1{+}k_2{\leq}\infc}\\
  \quad \qquad \qquad \code{\vee k_1{>}0{\wedge}k_2{\leq}0{\wedge}k_1{+}k_2{\geq}{-\infc}}\\
  \quad \qquad \qquad \code{\vee k_1{\leq}0{\wedge}k_2{>}0{\wedge}k_1{+}k_2{\geq}{-\infc}}\\
  \quad \qquad \qquad \code{\vee k_1{\leq}0{\wedge}k_2{\leq}0{\wedge}(k_1{+}k_2{\geq}{-\infc}{\vee}k_1{=}0)}\\
  \quad \code{ensures~res{=}k_1{+}k_2;}\\
 \end{array}
\]

The specification of these methods (\code{uadd} and \code{add}) uses the infinite
constants (\code{\infc} and \code{-\infc}) from the enriched specification
language given in the previous section. An expressive specification language 
reduces the task of integer overflow checking to just specifying and verifying of appropriate
methods.
After translation of basic operators into method calls, during forward verification
the \code{\verirulen{CALL}} rule will ensure that we check each operation for integer overflow.
Thus a simple encoding of basic operators and translation of the source program before verification 
enables us to do integer overflow checking along with automated verification.

\section*{Soundness}\label{sec:soundness}
In this section we outline the soundness properties of our entailment 
procedure with infinities and the forward verifier with 
integer overflow checking. We assume the soundness of the underlying
entailment checker and verifier \cite{SCP12:Chin}.

\begin{lemma} 
\emph{(Equisatisfiable Normalization)}\\
\label{lemma:norm}
If $\pure \leadsto \pure_{norm}$ 
then $SAT(\pure) \implies SAT(\pure_{norm})$ 
and $SAT(\pure_{norm}) \implies SAT(\pure)$ 
\end{lemma}

\begin{proof}
We sketch the proof for each normalization rule given in
figure \ref{fig.normalize}.

case $\rulen{NORM-INF-INF}$: From the first rule we get,  
$SAT(\code{\infc} = \code{\infc}) \equiv \true $\\
and the normalization gives $\code{\infc} = \code{\infc} \leadsto \true$, since $SAT(\true) \equiv \true$\\
we have, $SAT(\code{\infc} = \code{\infc}) \implies SAT(\code{\infc} = \code{\infc} \leadsto \true)$\\
and $SAT(\code{\infc} = \code{\infc} \leadsto \true) \implies SAT(\code{\infc} = \code{\infc})$.\\
Hence the normalization preserves satisfiability of pure formulas. We
can prove the other rules in $\rulen{NORM-INF-INF}$ similarly.

case $\rulen{NORM-CONST-INF}$: From the first rule we get,
$SAT(k = \code{\infc}) \equiv \false $\\
and the normalization gives $k = \code{\infc} \leadsto \false$, since $SAT(\false) \equiv \false$\\
we have, $SAT(k = \code{\infc}) \implies SAT(k = \code{\infc} \leadsto \false)$\\
and $SAT(k = \code{\infc} \leadsto \false) \implies SAT(k = \code{\infc})$.\\
Hence the normalization preserves satisfiability of pure formulas. We
can prove the other rules in $\rulen{NORM-CONST-INF}$ similarly.

case $\rulen{NORM-VAR-INF}$: We sketch the proof for the following rule,\\
$SAT(\code{\infc} \leq v)$\\
$\iff SAT(\code{\infc} < v \vee \code{\infc} = v)$\\
$\iff SAT(\false \vee \code{\infc} = v)$\\
we have, $SAT(\code{\infc} \leq v) \implies SAT(\code{\infc} \leq v \leadsto v = \code{\infc})$ \\
and $SAT(\code{\infc} \leq v \leadsto v = \code{\infc}) \implies SAT(\code{\infc} \leq v)$\\
Hence the normalization preserves satisfiability of pure formulas.
Other rules from $\rulen{NORM-VAR-INF}$ can be proven similarly.

case $\rulen{NORM-MIN-MAX}$: We sketch the proof for the following rule,\\
$SAT(max(a,\code{\infc}))$\\
$\iff SAT((a > \code{\infc} \implies a) \vee (a \leq \code{\infc} \implies \code{\infc}))$\\
$\iff SAT((\false \implies a) \vee (a \leq \code{\infc} \implies \code{\infc}))$\\
$\iff SAT((\true) \vee (a \leq \code{\infc} \implies \code{\infc}))$\\
$\iff SAT(a \leq \code{\infc} \implies \code{\infc})$\\
$\iff SAT(\true \implies \code{\infc})$\\
$\iff SAT(\code{\infc})$\\
we have, $SAT(max(a,\code{\infc})) \implies SAT(max(a,\code{\infc}) \leadsto \code{\infc})$ \\
and $SAT(max(a,\code{\infc}) \leadsto \code{\infc}) \implies SAT(max(a,\code{\infc}))$\\
Hence the normalization preserves satisfiability of pure formulas.
Other rules from $\rulen{NORM-MIN-MAX}$ can be proven similarly.
\qed
\end{proof}

\begin{lemma}
\emph{(Soundness of Integer Overflow Checking)}\\
\label{lemma:ioc}
If the program $e$ has an integer overflow ($\mustF$) then,\\
with forward verification $\vdash\{\D_1\}e\{\rn\}$, we have $(\true)\mustF \in \rn$
\end{lemma}

\begin{proof}
Provided all basic operators on integers in the program are translated to method calls that
check for integer overflows. The soundness of integer overflow checking 
follows from lemma \ref{lemma:norm} and the 
soundness of error calculus \cite{NFM13:Loc}. \qed
\end{proof}

The soundness of heap entailment
and forward verification with separation logic based
specifications is already established in \cite{SCP12:Chin}.
Lemma \ref{lemma:norm} establishes that the normalization rules indeed
preserve the satisfiability of pure formulas. Lemma \ref{lemma:ioc} then
shows that the integer overflow checking with forward verification is 
sound. If the program has an integer overflow the forward verification
with integer overflow checking detects it.

\section*{Experiments}
\label{sec:impl}
We have implemented our approach in an OCaml prototype called \code{HIPioc}
\footnote{Available at \url{http://loris-5.d2.comp.nus.edu.sg/SLPAInf/SLPAInf.ova} (md5sum 4afb66d65bfa442726717844f46eb7b6)}
evaluate  automated verification using logical infinities (\code{\infc})
we created benchmark of several programs that use infinite constants as
sentinel values in searching and sorting. As an example the following predicate definition of 
a sorted linked list
uses \code{\infc} in the base case to express that the minimum value in 
an empty list is infinity. 

\[
 \begin{array}{l}
 \code{data~\{int~val;node~next\}}\\
 \\
 \code{\view{Sortedll}{\self,min} {\veq}}
 \code{(\self{=}\nil {\wedge} min{=}\code{\infc}) {\vee}}\\ 
 \quad \code{\,\exists \, q \cdot (\hformn{\self}{node}{min,q}{\sep} \hformp{q}{Sortedll}{minrest}) 
 {\wedge} min {<} minrest}
 \end{array}
\]

In addition, \code{HIPioc} allows us to do integer overflow checking of programs during verification.
We have run \code{HIPioc} on several existing verification benchmarks which
contain different kinds of programs. 
The benchmarks include many examples of programing manipulating 
complex heap, arrays, concurrency (barriers  and variable permissions)
and some programs taken from real software (SIR). The results are shown in the table below.
$Sorting (with~\infc)$ are the programs which use \code{\infc} as sentinel value in predicate 
definitions as described above and do not contain integer overflows. Comparing the times between \code{HIPioc} 
and previous version we see that the verification with \code{\infc} in  general
adds some overhead.

\[
\begin{array}[t]{|c|c|c|c|c|c|c|}
\hline
Benchmark & LOC & Num~of & Time & Time & Integer & False \\
Programs & (Total) & Programs & (Secs) & (HIPioc) & Overflows & Positives\\
\hline
Sorting (with~\infc) & 282 & 4 & 5.45 & 5.42 & 0 & 0\\
Arrays & 1432 & 21 & 47.92 & 76.65 & 1 & 0\\
HIP/SLEEK~$\cite{SCP12:Chin}$ & 5779 & 42 & 56.15 & 78.80 & 4 & 0\\
Imm~$\cite{OOPSLA11:David}$ & 2069 & 11 & 120.82 & 126.61 & 18 & 0\\
VPerm~$\cite{ICFEM12:Khanh}$ & 778 & 14 & 3.43 & 3.46 & 3 & 0\\
Barriers~$\cite{LMCS12:Cristian}$ & 1281 & 10 & 60.54 & 60.83 & 16 & 0\\
SIR~$\cite{NFM13:Loc}$ & 2616 & 4 & 34.64 & 41.73 & 1 & 1\\
\hline
Total & 14237 & 106 & 328.95 & 393.5 & 43 & 1\\
\hline
\end{array}
\]

In total we found 43 integer overflows in these programs. Since
these programs were already verified by using automated provers we 
notice that integer overflows are prevalent even in verified software.
As our verification system is based on over approximation (sound but not complete), 
it can lead to false positives when finding integer overflows. In practice,
we see that only 1 example in the experiments we conducted contained a false positive.
The time taken to do verification of programs without and with integer overflow checking shows that
our technique can be applied to do integer overflow checking of programs during 
verification with modest overhead.

\section*{Related Work}
\label{sec:related}
There has been considerable interest in recent years to detect and prevent
integer overflows in software \cite{Cotroneo12}. Dietz et al. \cite{Dietz:ICSE12} present a study of integer
overflows in C/C++ programs. They find out that intentional and unintentional integer overflows
are prevalent in system software. 
Integer overflows often lead to 
exploitable vulnerabilities in programs \cite{Chirstey:SANS11}.
In this paper we presented a method to detect
unintentional integer overflows and provided a mechanism to specify intentional integer overflows.
Program transformations \cite{Coker:ICSE13} can be used to guide the programmer and 
aid in refactoring the source code to avoid
integer overflows. Our focus is on specification of intentional integer overflows which
helps make the conditions under which the program may use an integer overflow explicit.
It also aids in automated verification as such cases can be validated as 
error scenarios for the program.

Most existing techniques for detecting integer overflows are focused on dynamic checking and 
testing of programs \cite{Molnar:USENIX2009,NDSS2009:Wang,ARES2009:Chen,NDSS2007:Brumley}. Dynamic analysis suffers from the path explosion problem and although several improvements in constraints solving have been proposed \cite{sharma2012critical, sharma2013empirical} the approach cannot guarantee the absence of integer overflows. 
There are not many verification or static analysis tool that can do integer overflow checking.
KINT \cite{OSDI2012:Wang} is a static analysis tool which can detect integer overflows 
by solving constraints generated
from source code of programs. 
Another static analysis based approach by Moy et al. \cite{MSR2009:Moy} uses the Z3 solver to 
do integer overflow checking as part of the PREFIX tool \cite{PREFIX}. A certified prover for presburger arithmetic extended with positive and negative infinities has been described in \cite{CRAV2015, sharma2015certified}. 

Our focus in this paper is on integrating integer overflow checking
with program verification to improve the reliability of verified software.
Dynamic techniques  may not explore all paths in the programs while static techniques suffer form loss of precision in tracking integer overflows.
Our specification mechanism allows us to integrate integer overflow checking inside a
prover for specification logic. This allows us to track integer overflows through
the heap and inside various data structures. The benefit of this integration is that
we can detect numeric integer overflow errors in programs with complex sharing
and heap manipulation.

\section*{Conclusion}
\label{sec.conc}
Integer overflows are a major source of errors in programs.
Most verification systems do not focus on the underlying numeric operations on integers
and do not handle integer overflow checking.
We presented a technique to do integer overflow checking of programs during verification.
Our specification mechanism also allows expressing intentional uses of integer overflows.
We implemented a prototype of our proposal inside an existing verifier and found
real integer overflow bugs in benchmarks of verified software.

\bibliography{all}

\end{document}